%% file: main.tex
\documentclass[letterpaper, 11pt]{article}
\usepackage{fullpage}
\usepackage{graphicx}
\usepackage{amsmath, amsthm, amssymb, cite, color, comment, mathtools, url}
\usepackage[unicode]{hyperref}
\usepackage{algorithm}

\newtheorem{theorem}{Theorem}[section]
\newtheorem{lemma}[theorem]{Lemma}

\theoremstyle{definition}
\newtheorem{definition}[theorem]{Definition}

\theoremstyle{remark}
\newtheorem{remark}[theorem]{Remark}

\newcommand{\RR}{{\mathbb R}}

\newcommand{\cX}{\mathcal{X}}

\newcommand{\tmu}{{\tilde{\mu}}}

\newcommand{\tp}{{\tilde{p}}}
\newcommand{\tw}{{\tilde{w}}}

\title{Fast Primal-Dual Update against Local Weight Update\\in Linear Assignment Problem and Its Application}
\author{
  Kohei Morita\thanks{Indeed Japan, Tokyo, Japan. Email: \texttt{k.morita06@gmail.com}} \and
  Shinya Shiroshita\thanks{Preferred Networks, Tokyo, Japan. Email: \texttt{shiroshita@preferred.jp}} \and
  Yutaro Yamaguchi\thanks{Osaka University, Osaka, Japan. Email: \texttt{yutaro.yamaguchi@ist.osaka-u.ac.jp}} \and
  Yu Yokoi\thanks{Tokyo Institute of Technology, Tokyo, Japan. Email: \texttt{yokoi@c.titech.ac.jp}}}
\date{\empty}

\begin{document}
\maketitle
\thispagestyle{empty}

\begin{abstract}
We consider a dynamic situation in the weighted bipartite matching problem: edge weights in the input graph are repeatedly updated and we are asked to maintain an optimal matching at any moment.
A trivial approach is to compute an optimal matching from scratch each time an update occurs.
In this paper, we show that if each update occurs locally around a single vertex, then a single execution of Dijkstra's algorithm is sufficient to preserve optimality with the aid of a dual solution.
As an application of our result, we provide a faster implementation of the envy-cycle procedure for finding an envy-free allocation of indivisible items.
Our algorithm runs in $\mathrm{O}(mn^2)$ time, while the known bound of the original one is $\mathrm{O}(mn^3)$, where $n$ and $m$ denote the numbers of agents and items, respectively.
\end{abstract}


\clearpage
\setcounter{page}{1}

\input{introduction.tex}
\input{preliminaries.tex}
\input{PDU.tex}
\input{ECP.tex}

\section*{Acknowledgments}
The proposed algorithms were developed via the preparation of Problem K of the International Collegiate Programming Contest (ICPC) 2021 Asia Yokohama Regional Contest \cite{ICPC2021}.
We are grateful to all of the contest organizers and participants.

This work was supported by JSPS KAKENHI Grant Numbers 20K19743 and 20H00605, and JST PRESTO Grant Number JPMJPR212B.

\end{document}

%% file: introduction.tex

\section{Introduction}
The weighted matching problem in bipartite graphs is a fundamental combinatorial optimization problem.
Initiated by the so-called Hungarian method~\cite{Kuhn1955}, a variety of efficient algorithms for the problem have been developed \cite{BGK1977, EK1972, HR1980, Iri1960, Tomizawa1971}.

We consider a dynamic situation in this problem: we are given an edge-weighted bipartite graph in which edge weights are repeatedly updated, and asked to maintain an optimal matching at any moment. 
While many papers on dynamic matching consider fast approximation for the situation in which edges arbitrarily appear and disappear (one by one), we focus on maintaining an exact optimal solution when the underlying graph is fixed and edges change their weights (possibly simultaneously).\footnote{In a similar context, Sankowski~\cite{Sankowski2007} proposed pseudopolynomial-time exact maintenance of the maximum weight of a matching in a dynamic bipartite graph. For the literature of usual dynamic graph algorithms including those for matchings, see, e.g., \cite{BDL2021,HHS2021,Henzinger2018,SW2017}.}

A trivial approach is to compute an optimal matching from scratch each time the weights are updated.
This simple method is applicable to any update; using a typical efficient algorithm, it requires $\mathrm{O}(n(m + n \log n))$ time per update, where $n$ and $m$ denote the numbers of vertices and edges, respectively.

In this paper, we show that if each update occurs around a single vertex, a more efficient algorithm can be designed.
Specifically, in such a situation, by maintaining a dual optimal solution simultaneously, it suffices to perform Dijkstra's algorithm only once for updating both primal and dual solutions.\footnote{The readers familiar with the Hungarian method may think that the idea is obvious and moreover it is a folklore. We, however, could not find any description of this fact, and we believe that it is worth reporting with a clearly described proof.}
This requires $\mathrm{O}(m + n \log n)$ time per update.

We first describe our idea in the setting of the assignment problem, i.e., a special case in which the input bipartite graph is balanced and we are asked to find an optimal perfect matching.
Suppose that the edge weights are updated around a vertex $s$.
On the primal side, we can obtain a new optimal matching from an old one by exchanging the edges along a minimum-weight cycle intersecting $s$ in the auxiliary graph.
Such a cycle can be found by computing shortest paths from $s$ in the modified auxiliary graph obtained by flipping the matching edge incident to $s$.
On the dual side, from the shortest path length, we can obtain a new potential with respect to the new optimal matching.
Since the modified auxiliary graph has no negative cycles, these primal/dual updates can be done in $\mathrm{O}(nm)$ time by the Bellman--Ford algorithm.
Our main observation is that when the weights are updated around a single vertex, the old potential can be easily adjusted to the modified auxiliary graph; this enables us to apply Dijkstra's algorithm, which requires $\mathrm{O}(m + n \log n)$ time.

We then discuss cases in which the input bipartite graph is unbalanced or the perfectness of a matching is not imposed.
Although such cases can be reduced to the assignment problem by adding dummy vertices and edges, a na\"ive reduction may essentially increase the input size.
We show that by avoiding an explicit reduction, the same computational time bound can be obtained.

We also demonstrate that our technique is useful for obtaining a faster implementation of the envy-cycle procedure for finding an \emph{EF1 (envy-free up to one item)} allocation of indivisible items \cite{Lipton2004}.
We observe that the envy-cycle procedure can be modified so that the part of cancelling envy cycles is formulated as the assignment problem, whose input graph is updated around a single vertex in each iteration.
Our primal-dual update algorithm fits this situation.
While the well-known running time bound of the original envy-cycle procedure is $\mathrm{O}(mn^3)$, our variant runs in $\mathrm{O}(mn^2)$ time, where $n$ and $m$ denote the numbers of agents and items, respectively.

The rest of the paper is organized as follows.
In Section~\ref{sec:preliminaries}, we provide the necessary definitions and basic facts on the assignment problem.
In Section~\ref{sec:PDU}, we describe our algorithm for updating primal/dual solutions in response to weight update around a vertex in the assignment problem, and discuss equivalent formulations of the weighted bipartite matching problem.
In Section~\ref{sec:improved}, we present an improved variant of the envy-cycle procedure.

%% file: preliminaries.tex

\section{Preliminaries}\label{sec:preliminaries}
In the \emph{(linear) assignment problem}, we are given an edge-weighted bipartite graph and asked to find a perfect matching of minimum total weight.
We briefly review the theory behind typical primal-dual approaches for this problem (see \cite{Schrijver2003} for more details).
Throughout the paper, except for Section~\ref{sec:unbalanced}, we assume that the input graph is simple and has a perfect matching.

The input is a pair $(G, w)$ of a bipartite graph $G = (V^+, V^-; E)$ with vertex set $V$ partitioned into $V^+$ and $V^-$ and edge set $E$, where $|V^+|=|V^-|=n$ and $|E|=m$, and edge weights $w \colon E \to \mathbb{R}$.\footnote{It is often assumed that the input graph is a complete bipartite graph. The present formulation reduces to this setting by adding all the absent edges with sufficiently large weight. This reduction, however, does not preserve the input size if the original graph is not dense, i.e., $m = \mathrm{o}(n^2)$. In this sense, we strictly evaluate the computational complexity.}
We also regard $G$ as a directed graph in which each edge is directed from $V^+$ to $V^-$.

A \emph{matching} $\mu\subseteq E$ in $G$ is an edge set in which every vertex appears as an endpoint at most once, and $\mu$ is called \emph{perfect} if all the vertices appear exactly once.
We also regard a perfect matching $\mu$ as a bijection from $V^+$ to $V^-$ and vice versa.
That is, we write $\mu(u) = v$ and $\mu^{-1}(v) = u$ if $e = (u, v) \in \mu$.
We define the weight of a matching $\mu$ as $w(\mu) = \sum_{e \in \mu}w(e)$.

Let $\mu$ be a matching in $G = (V, E)=(V^+, V^-; E)$.
The \emph{auxiliary graph} for $\mu$, denoted by $G_{\mu} = (V, E_{\mu})$, is a directed graph that is obtained from $G$ by flipping the direction of the edges in $\mu$, i.e.,
\begin{align*}
E_{\mu} \coloneqq (E \setminus \mu) \cup \bar{\mu} = \{\, e \mid e \in E \setminus \mu \,\} \cup \{\, \bar{e} \mid e \in \mu \,\},
\end{align*}
where $\bar{e}$ denotes the reverse edge of $e$.
We also define the \emph{auxiliary weight} $w_{\mu} \colon E_{\mu} \to \mathbb{R}$ by
\[w_{\mu}(e) \coloneqq \begin{cases}
w(e) & (e \in E \setminus \mu),\\
-w(\bar{e}) & (e \in \bar{\mu}).
\end{cases}\]
We define the weight of a path or a cycle in $G_\mu$, say $P$, as $w_\mu(P) = \sum_{e \in P} w_\mu(e)$.
We call a path/cycle \emph{shortest} if it is of minimum weight (under some specified condition) and \emph{negative} if it is of negative weight.

Let $V_{\mu}^+$ and $V_{\mu}^-$ denote the sets of vertices not matched by $\mu$ in $V^+$ and in $V^-$, respectively.
An \emph{augmenting path} in $G_\mu$ is a directed path starting from a vertex in $V_{\mu}^+$ and ending at a vertex in $V_{\mu}^-$.
Let $P$ be an augmenting path or a cycle in $G_\mu$.
We then obtain another matching, denoted by $\mu \triangle P$, from $\mu$ by exchanging the edges in $E \setminus \mu$ and in $\mu$ along $P$, i.e., by adding the edges in $(E \setminus \mu) \cap P$ and by removing the edges in $\mu \cap P$.
Note that $w(\mu \triangle P) = w(\mu) + w_\mu(P)$ by definition.

A celebrated feature of a minimum-weight matching, which is utilized in many algorithms, is described below.

\begin{lemma}[cf.~{\cite[Section 17.2]{Schrijver2003}}]\label{lem:extreme}
Let $\mu$ be a matching of size $k$.
\begin{itemize}
\setlength{\itemsep}{.5mm}
\item If $\mu$ is a minimum-weight matching in $G$ subject to $|\mu| = k$, then $(G_{\mu}, w_{\mu})$ has no negative cycles; these conditions are equivalent when $k = n$.
\item If $\mu$ is a minimum-weight matching in $G$ subject to $|\mu| = k$ and $P$ is a shortest augmenting path in $(G_{\mu}, w_{\mu})$, then $\mu' \coloneqq \mu \triangle P$ is a minimum-weight matching in $G$ subject to $|\mu'| = k + 1$.
\end{itemize}
\end{lemma}

By Lemma~\ref{lem:extreme}, a minimum-weight matching $\mu$ of size $k$ can be augmented to a minimum-weight matching $\mu'$ of size $k + 1$ by finding a shortest path in a directed graph with no negative cycles.
This is done in $\mathrm{O}(nm)$ time by the Bellman--Ford algorithm.
Furthermore, this is done in $\mathrm{O}(m + n \log n)$ time by Dijkstra's algorithm (with the aid of the Fibonacci heap) if we simultaneously maintain a potential (i.e., a dual solution), which is defined as follows.

For a matching $\mu$ and a function $p \colon V \to \mathbb{R}$, we define the \emph{reduced weight} $w_{\mu, p} \colon E_{\mu} \to \mathbb{R}$ by
\[w_{\mu, p}(e) \coloneqq w_{\mu}(e) + p(u) - p(v) \quad (e = (u, v) \in E_{\mu}).\]
We say that $p$ is a \emph{potential for $(G_\mu, w_\mu)$} if $w_{\mu, p}(e)\geq 0$ for every $e\in E_\mu$.

We here summarize the Hungarian method with this speed-up.
First, we set $\mu_0 \leftarrow \emptyset$ and
\[p_0(v) \leftarrow \begin{cases}
	0 & (v \in V^+),\\
  	\min\{\, w(e) \mid e \in \delta_G(v) \,\} & (v \in V^-),
\end{cases}\]
where $\delta_G(v)$ denotes the set of edges incident to $v$ in $G$.
Then, for each $k = 0, 1, \dots, n - 1$, we successively compute a minimum-weight perfect matching $\mu_{k+1}$ of size $k + 1$ and a potential $p_{k+1}$ for $(G_{\mu_{k+1}}, w_{\mu_{k+1}})$ using $\mu_k$ and $p_k$ as follows.

In each iteration, we compute a shortest path from $V_{\mu_k}^+$ to each vertex $v \in V$ in $(G_{\mu_k}, w_{\mu_k})$ by Dijkstra's algorithm with the aid of the potential $p_k$.
On the primal side, let $P$ be a shortest augmenting path, and set $\mu_{k+1} \leftarrow \mu_k \triangle P$.
On the dual side, for each vertex $v \in V$, let $d(v)$ denote the weight of a shortest path from $V_{\mu_k}^+$ to $v$, and set $p_{k+1}(v) \leftarrow d(v)$ (where we set $p_{k+1}(v) \leftarrow p_k(v)$ if $v$ is not reachable from $V_{\mu_k}^+$).

The bottleneck is executing Dijkstra's algorithm $n$ times; hence, the computational time of the whole algorithm is bounded by $\mathrm{O}(n(m + n \log n))$.

%% file: PDU.tex

\section{Primal/Dual Update against Weight Update around a Vertex}\label{sec:PDU}
In this section, we propose an $\mathrm{O}(m + n \log n)$-time algorithm for updating an optimal matching and a corresponding potential when the edge weights are updated around a single vertex.
We show this for the assignment problem in Section~\ref{sec:balanced}, and discuss other polynomially equivalent formulations of the weighted bipartite matching problem in Section~\ref{sec:unbalanced}.

\subsection{Assignment Problem Case}\label{sec:balanced}
The problem under consideration is described below. 
Suppose we have a pair of primal and dual optimal solutions for an instance $(G, w)$ of the assignment problem, i.e.,  we have a minimum-weight perfect matching $\mu$  and a potential $p$ for $(G_\mu, w_\mu)$.
Now, the edge weights are updated arbitrarily but only around a single vertex $s \in V$, and we are asked to find a pair of primal and dual optimal solutions for the new instance $(G, \tw)$.
By symmetry, we can assume $s \in V^+$.
In addition, for the sake of simplicity, we assume that $\delta_G(s) = \{s\} \times V^-$ by adding at most $n$ absent edges with sufficiently large weights.

Our algorithm is shown in Algorithm~\ref{alg:PDU}, which works as follows. 
Let $e'=(s,t)$ be the edge in $\mu$ incident to $s$ and let $\mu'\coloneqq \mu\setminus\{e'\}$.  
Then, in $G_{\mu'}$, all the edges in $\delta_G(s) = \{s\} \times V^-$ are directed from $s$ to $V^{-}$. 
Our objective is to find shortest paths from $s$ in $G_{\mu'}$ with respect to the new weights $\tw$. 
Note that $p$ is a potential for $(G_\mu, w_\mu)$ (i.e., $w_{\mu, p}(e)\geq 0$ for every $e\in E_\mu$), but it may not be a potential for the current graph $(G_{\mu'}, \tw_{\mu'})$.
However, since $(G_{\mu'}, \tw_{\mu'})$ differs from $(G_\mu, w_\mu)$ only around $s$, we can recover the nonnegativity condition only by modifying the value of $p$ on $s$.
Specifically, we define $p'$ as $p'(s)\coloneqq -\min\{\, \tw(e)-p(v)\mid  e=(s,v)\in \delta_{G}(s)\,\}$ and $p'(v)\coloneqq p(v)$ for every $v\in V\setminus\{s\}$. Then, $p'$ is indeed a potential for $(G_{\mu'}, \tw_{\mu'})$, and we can efficiently compute a shortest path from $s$ to each vertex by Dijkstra's algorithm.
Using the obtained paths and their weights, we update the primal and dual solutions.

\begin{algorithm}[h]
	\caption{~Primal Dual Update}
	\begin{description}
		\setlength{\itemsep}{-1mm}
		\item[Input:] A bipartite graph $G = (V^+, V^-; E)$, edge weights $w \colon E \to \mathbb{R}$, a minimum-weight perfect matching $\mu$ in $(G, w)$, a potential $p$ for $(G_\mu, w_\mu)$, a vertex $s \in V^+$ with $\delta_G(s) = \{s\} \times V^-$, and new edge weights $\tw \colon E \to \mathbb{R}$ such that $\tw(e) = w(e)$ for every edge $e \in E \setminus \delta_G(s)$.
		\item[Output:] A minimum-weight perfect matching $\tmu$ in $(G, \tw)$ and a potential $\tp$ for $(G_{\tmu}, \tw_{\tmu})$.
	\end{description}
	\begin{enumerate}
		\setlength{\itemsep}{0mm}
		\setlength{\leftskip}{-1mm}
		\item
		Let $e' = (s, t)$ be the edge in $\mu$ incident to $s$  and let $\mu' = \mu \setminus \{e'\}$.
		\item
		Set $p'(s)=-\min\{\, \tw(e)-p(v)\mid  e=(s,v)\in \delta_{G}(s)\,\}$ and $p'(v)=p(v)$  for every \mbox{$v\in V\setminus\{s\}$}. 
		Compute a shortest path $P_v$ from $s$ to each $v \in V$ in $(G_{\mu'}, \tw_{\mu'})$ by Dijkstra's algorithm with the aid of potential $p'$.
		\item
		Set $\tmu \leftarrow \mu' \triangle P_t$. 
		For each $v \in V$, set $\tp(v) \leftarrow \tilde{w}_{\mu'}(P_v)$.
		\item
		Return $\tmu$ and $\tp$.
	\end{enumerate}
	\label{alg:PDU}
\end{algorithm}

\begin{theorem}\label{thm:PDU}
Algorithm~\ref{alg:PDU} correctly finds $\tmu$ and $\tp$ in $\mathrm{O}(m + n \log n)$ time.
\end{theorem}

\begin{proof}
The algorithm clearly runs in $\mathrm{O}(m + n \log n)$ time.
We prove its correctness.

We first confirm that $p'$ in Step 2 is indeed a potential for $(G_{\mu'}, \tw_{\mu'})$.
Since $p$ is a potential for $(G_\mu, w_\mu)$, 
we have $\tw_{\mu', p'}(e) = w_{\mu, p}(e) \ge 0$ for every edge $e \in E_{\mu'} \setminus \delta_{G}(s)\subseteq E_{\mu}$.
Also, as every edge $e = (s, v) \in \delta_G(s)$ is directed from $s$ to $v$ in $G_{\mu'}$, we have
\begin{align*}
\tw_{\mu', p'}(e) &= \tw(e) + p'(s) - p'(v)\\
&= \tw(e) - p(v)-\min \left\{\,  \tw(e^*)-p(v^*)\mid  e^* = (s,v^*)\in \delta_{G}(s)\,\right\}\\
&\ge 0.
\end{align*}

We next show the correctness of the primal update.
Since $e'=(s,t)$ itself forms an $s$--$t$ path in $G_{\mu'}$, the weight of $P_t$ is at most $\tw(e')$; hence, $\tw(\tmu)\leq \tw(\mu')+\tw(e')=\tw(\mu)$. 
Suppose to the contrary that there is a perfect matching $\mu^*$ in $G$ with $\tw(\mu^*)<\tw(\tmu)\leq \tw(\mu)$.
Then, $\mu^* \triangle \mu$ contains a negative cycle $C$ in $(G_\mu, \tw_\mu)$ (where each edge in $\mu$ is flipped).
Since $(G_\mu, w_\mu)$ has no negative cycles (by Lemma~\ref{lem:extreme}) and $\tw$ differs from $w$ only around $s$, the cycle $C$ must intersect $s$.
This implies that $C$ is the only negative cycle in $\mu^* \triangle \mu$ and hence $\tw_\mu(C) \le \tw(\mu^*) - \tw(\mu)$.
In addition, as $\bar{e}'$ is the only edge entering $s$ in $G_{\mu}$, it is contained in $C$.
By removing $\bar{e}'$ from $C$, we obtain an $s$--$t$ path $P^*_t$ in $(G_\mu, \tw_{\mu'})$ such that $\tw_{\mu'}(P^*_t)=\tw_{\mu}(C)-\tw_{\mu}(\bar{e}')\le\tw(\mu^*)-\tw(\mu)+\tw(e')<\tw(\tmu)-\tw(\mu')=\tw_{\mu'}(P_t)$, which contradicts the fact that $P_t$ is a shortest path.

Finally, we show that $\tp$ is a potential for $(G_{\tmu}, \tw_{\tmu})$.
Since $\tp(v)$ is the weight of the shortest $s$--$v$ path $P_v$ in $(G_{\mu'}, \tw_{\mu'})$ for any $v\in V$, the function $\tp$ is clearly a potential for $(G_{\mu'}, \tw_{\mu'})$.
That is, we have $\tw_{\mu', \tp}(e) = \tw_{\mu'}(e) + \tp(u) - \tp(v) \ge 0$ for each $e = (u, v) \in E_{\mu'}$.
In particular, the equality holds for every edge $e$ used in the shortest $s$--$v$ path $P_v$ for some $v\in V$.
Note that $G_{\tmu}$ is obtained from $G_{\mu'}$ by flipping the direction of the edges on $P_t$, whose $\tw_{\mu', \tp}$ values are $0$.
Therefore, $\tw_{\tmu, \tp}(e)\geq 0$ holds for every edge $e\in E_{\tmu}$, and hence $\tp$ is a potential for  $(G_{\tmu}, \tw_{\tmu})$.
\end{proof}

Using Algorithm~\ref{alg:PDU}, we can design a simple $\mathrm{O}(n(m + n \log n))$-time algorithm for the assignment problem as follows.
We first set $w'(e) \leftarrow 0$ for each $e \in E$, find any perfect matching $\mu$ in $G$, and initialize $p(v) \leftarrow 0$ for each $v \in V$.
This part is done in $\mathrm{O}(nm)$ time by a na\"{i}ve augmenting path algorithm (or faster by sophisticated algorithms, but this is not a bottleneck).
Clearly, $\mu$ is a minimum-weight perfect matching in $(G, w')$ and $p$ is a potential for $(G_\mu, w'_\mu)$.
We then repeatedly make $w'$ close to the actual input weights $w$ by updating it around a single vertex $s \in V^+$ in any order.

In each iteration, we set
\[\tw(e) \leftarrow \begin{cases}
	w(e) & (e \in \delta_G(s)),\\
	w'(e) & (\text{otherwise}),\\
\end{cases}\]
and find a minimum-weight perfect matching $\tmu$ in $(G, \tw)$ and a potential $\tp$ for $(G_\tmu, \tw_\tmu)$ by applying Algorithm~\ref{alg:PDU} to $(G, w', \mu, p, s, \tw)$, where we temporarily add the absent edges around $s$ with sufficiently large weights if necessary.
We then update $w' \leftarrow \tw$, $\mu \leftarrow \tmu$, and $p \leftarrow \tp$.

After the iterations are completed, we have $w' = w$.
The correctness and the computational time bound directly follow from Theorem~\ref{thm:PDU}.

\subsection{Other Formulations}\label{sec:unbalanced}
In this section, we discuss other polynomially equivalent formulations of our problem so that we can deal with various settings. 
Consider the situation in which we are required to find a maximum-weight (not necessarily perfect) matching in an edge-weighted (not necessarily balanced) bipartite graph.
First, by negating all the edge weights, the objective turns into minimizing the total weight of a matching.

Let $G = (V^+, V^-; E)$ be a bipartite graph with $|V^+| \ge |V^-| > 0$, and let $w \colon E \to \mathbb{R}$.
Without loss of generality, we assume that $w$ is nonpositive, since any positive-weight edge can be excluded from an optimal matching.
We say that a matching $\mu$ in $G$ is \emph{right-perfect} if all the vertices in $V^-$ appear exactly once.
By adding a vertex $v'$ to $V^+$ with an edge $(v', v)$ with weight $0$ for each $v \in V^-$, any instance can be transformed into an instance having a right-perfect optimal matching.
Note that the input size does not essentially increase.
In the following, let $(G, w)$ be the graph after this transformation, so we have $|V^+| > |V^-|$.
Let $n = |V^+|$ and $m = |E|$.

We now consider the following analogous problem: we have a minimum-weight right-perfect matching $\mu$ in an edge-weighted unbalanced bipartite graph $(G, w)$ and a potential $p$ for $(G_\mu, w_\mu)$, the edge weights $w$ change arbitrarily but only around a single vertex $s \in V$, and we are asked to find a minimum-weight right-perfect matching $\tmu$ in the new graph $(G, \tw)$ and a potential $\tp$ for $(G_\tmu, \tw_\tmu)$.

By adding $|V^+| - |V^-|$ dummy vertices to $V^-$ that are adjacent to all the vertices in $V^+$ by edges with weight $0$, this problem reduces to the problem solved in Section~\ref{sec:balanced}.
Let $(G', w')$ be the edge-weighted balanced bipartite graph obtained by this reduction, and let $V'$ and $E'$ denote its vertex set and edge set, respectively.
We then have $|V'| = 2n = \mathrm{O}(n)$, but $|E'|$ may not be bounded by $\mathrm{O}(m + n \log n)$ if $G$ is sparse and far from balanced.
Hence, this na\"{i}ve reduction may worsen the computational time bound in Theorem~\ref{thm:PDU}.

We show that this problem can be solved in $\mathrm{O}(m + n \log n)$ time by emulating Algorithm~\ref{alg:PDU} for $(G', w')$ with avoiding its explicit construction.
Let $\mu'$ and $p'$ be the corresponding optimal solutions and $(G', \tw')$ be the corresponding new graph in the reduced problem.

\begin{lemma}\label{lem:dummy}
If $(G'_{\mu'}, \tw'_{\mu'})$ contains a negative cycle, then there exists a shortest cycle that intersects at most one dummy vertex.
\end{lemma}

\begin{proof}
Let $C$ be a shortest cycle in $(G'_{\mu'}, \tw'_{\mu'})$ that consists of as few edges as possible.
We show that $C$ intersects at most one dummy vertex.
Suppose to the contrary that $C$ intersects two distinct dummy vertices $v_1$ and $v_2$ in the extended $V^-$.
Then, $C$ is divided into two paths, $P_1$ from $v_1$ to $v_2$ and $P_2$ from $v_2$ to $v_1$.
By Lemma~\ref{lem:extreme}, the negative cycle $C$ must intersect $s$; as $s$ is not a dummy vertex, we assume that $P_1$ intersects $s$ and $P_2$ does not without loss of generality.

Let $u \in V^+$ be the vertex just before $v_1$ on $P_2$.
If ${\mu'}(u)=v_2$, then $\tw'_{\mu'}(P_2) = \tw(u, v_2) + \tw(u, v_1) = 0$ by definition.
Otherwise, $G'_{\mu'}$ contains a directed edge $e = (u, v_2)$, and one can obtain a cycle $C'$ from $P_2$ by replacing the last edge $(u, v_1)$ with $e$. 
This $C'$ does not contain $s$, and hence $\tw'_{\mu'}(P_2) = \tw'_{\mu'}(C') \ge 0$ by Lemma~\ref{lem:extreme}.
Thus, in either case, we have $\tw'_{\mu'}(P_2) \ge 0$, which implies $\tw'_{\mu'}(P_1) < 0$.

Let $u' \in V^+$ be the vertex just before $v_2$ on $P_1$.
If ${\mu'}(u')=v_1$, then $\tw'_{\mu'}(P_1) = \tw(u', v_1) + \tw(u', v_2) = 0$, a contradiction.
Thus, there exists a directed edge $e' = (u', v_1)$ in $G'_{\mu'}$, and a cycle $C''$ with $\tw'_{\mu'}(C'') \le \tw'_{\mu'}(C)$ and $|C''| < |C|$ can be obtained from $C$ by replacing the subpath consisting of $(u', v_2)$ and $P_2$ with $e'$.
This contradicts the choice of $C$; thus, we are done.
\end{proof}

\begin{lemma}\label{lem:dummy2}
For any vertex $v \in V$, there exists a shortest $s$--$v$ path in $(G'_{\mu'}, \tw'_{\mu'})$ that intersects at most one dummy vertex.
Moreover, for all the dummy vertices, the weights of the shortest paths are the same.
\end{lemma}

\begin{proof}
The first claim is shown by almost the same argument as the proof of Lemma~\ref{lem:dummy}.
Pick a shortest $s$--$v$ path $P$ in $(G'_{\mu'}, \tw'_{\mu'})$ that consists of as few edges as possible.
If $P$ intersects two distinct dummy vertices $v_1$ and $v_2$, then the subpath between $v_1$ and $v_2$ can be skipped without increasing the weight, which contradicts the minimality of $P$.

To see the second claim, for two distinct dummy vertices $v_1$ and $v_2$, let $u_1$, $u_2$ be the vertices such that $(u_1, v_1), (u_2, v_2)\in \mu'$.
Then, by definition, $(G'_{\mu'}, \tw'_{\mu'})$ has four edges $(v_1, u_1)$, $(v_2, u_2)$, $(u_1, v_2)$, and $(u_2, v_1)$ of weight $0$, and hence we have $d(v_2) \le d(v_1)$ and $d(v_1) \le d(v_2)$.
\end{proof}

Lemmas~\ref{lem:dummy} and \ref{lem:dummy2} imply that, in order to emulate Algorithm~\ref{alg:PDU} for $(G', w')$, it suffices to introduce only one dummy vertex $x$ that is adjacent to all the vertices in $V^+$ with edges of weight $0$ and to modify the perfect matching constraint so that $x$ appears exactly $|V^+| - |V^-|$ times.
Thus, we have completed the proof.

%% file: ECP.tex

\section{An Improved Variant of Envy-Cycle Procedure}\label{sec:improved}
In this section, as an application of our algorithm, we present an improved variant of the \emph{envy-cycle procedure} for finding an envy-free allocation of indivisible items.
We describe the problem setting and the envy-cycle procedure in Section~\ref{sec:EFA}, and demonstrate how to improve its computational time via the assignment problem in Section~\ref{sec:improved2}.

\subsection{Envy-Free Allocation of Indivisible Items}\label{sec:EFA}
Consider a situation in which $n$ agents share $m$ indivisible items.
Let $N = [n] \coloneqq \{1, 2, \dots, n\}$ and $M = [m]$ be the sets of agents and items, respectively.
Each item subset $X \subseteq M$ is called a \emph{bundle}.
Each agent $i \in N$ has a \emph{valuation} $v_i \colon 2^M \to \mathbb{R}_{\ge 0}$ meaning that agent $i$ evaluates a bundle $X \subseteq M$ by a nonnegative real $v_i(X)$. We assume that $v_i$ is \emph{monotone}, i.e., $X \subseteq Y \subseteq M \implies v_i(X) \le v_i(Y)$.
In addition, $v_i$ is said to be \emph{additive} if $v_i(X) = \sum_{k \in X} v_{i, k}$ for every bundle $X \subseteq M$, where $v_{i, k} \coloneqq v_i(\{k\})$ $(k \in M)$.

An \emph{allocation} is an $N$-indexed subpartition $(X_i)_{i \in N}$ of $M$, i.e., $X_i \subseteq M$ for each $i\in N$ and $X_{i} \cap  X_{j} = \emptyset$ for distinct $i, j\in N$.
In an allocation $(X_i)_{i \in N}$, an agent $i \in N$ \emph{envies} another agent $j \in N$ if $v_{i}(X_{i}) < v_{i}(X_{j})$.
Since each item is indivisible, it is usually impossible to achieve a completely envy-free allocation of the whole item set $M$.
In the following, we describe two well-studied relaxations of envy-freeness \cite{CKMPSW2019, Lipton2004}.

\begin{definition}\label{def:EF1}
Let $\cX = (X_i)_{i \in N}$ be an allocation.
\begin{itemize}
\item $\cX$ is \emph{envy-free up to one item} (\emph{EF1} for short) if, for every pair of agents $i, j \in N$, we have $v_{i}(X_{i}) \ge v_{i}(X_{j})$ or $v_{i}(X_{i}) \ge v_{i}(X_{j} \setminus \{k\})$ for some $k \in X_{j}$.
\item $\cX$ is \emph{envy-free up to any item} (\emph{EFX} for short) if, for every pair of agents $i, j \in N$, we have $v_{i}(X_{i}) \ge v_{i}(X_{j} \setminus \{k\})$ for every $k \in X_{j}$.
\end{itemize}
\end{definition}

Clearly, an EFX allocation is EF1.
It is known that an EF1 allocation always exists and can be found by a rather simple algorithm, the so-called \emph{envy-cycle procedure} \cite{Lipton2004}, whereas the existence of an EFX allocation is known for very restricted situations and is widely open in general (see, e.g., \cite{ALMW2022}).

The envy-cycle procedure is summarized as follows.
We initialize $\cX \leftarrow (X_i)_{i \in N}$ with $X_i = \emptyset$ for each $i \in N$, and distribute the items one-by-one in any order.
In each iteration, we consider the \emph{envy graph} $H(\cX) = (N, F(\cX))$ for the current allocation $\cX$, which is defined by
\[F(\cX) \coloneqq \{\, (i, j) \mid i, j \in N,~v_{i}(X_{i}) < v_{i}(X_{j}) \,\}.\]
While $H(\cX)$ contains a cycle, we update $\cX$ by exchanging the bundles along it; that is, we pick a cycle $C$ and update $X_{i} \leftarrow X_{j}$ for all the edges $(i, j) \in C$ simultaneously.\footnote{This exchange is usually repeated until some agent has no incoming edges, and it suffices. We employ the present form to make its behavior analogous to the improved variant of the envy-cycle procedure described in the next section.}
Subsequently, as $H(\cX)$ has no cycles, there exists an agent $i \in N$ having no incoming edges.
To such an agent $i$, we distribute an unallocated item, say $k \in M\setminus \bigcup_{j\in N}X_j$; that is, we set $X_i \leftarrow X_i \cup \{k\}$.

It is easy to observe that at any point of this algorithm, $\cX$ is an EF1 allocation of the set of items distributed so far.
The computational time is bounded by counting the number of edges in $H(\cX)$, which can increase only when an item is distributed and must decrease when the bundles are exchanged along a cycle.

\begin{theorem}[Lipton et al.~{\cite[Theorem~2.1]{Lipton2004}}]\label{thm:Lipton}
The envy-cycle procedure correctly finds an EF1 allocation in $\mathrm{O}(mn^3)$ time.
\end{theorem}

In \cite{BK2020}, it was pointed out that if the valuations are additive and the values of items are in the same order for all the agents, this algorithm can be employed to find an EFX allocation.

\begin{theorem}[Barman and Krishnamurthy~{\cite[Lemma~3.5]{BK2020}}]\label{thm:Lipton2}
Suppose that the valuations of all agents are additive, and $v_{i,1} \ge v_{i,2} \ge \dots \ge v_{i,m}$ for every agent $i \in N$.
Then, if the items are distributed in ascending order of $k = 1, 2, \dots, m$, the envy-cycle procedure correctly finds an EFX allocation in $\mathrm{O}(mn^3)$ time.
\end{theorem}

\begin{remark}
When the valuations are additive, an EF1 allocation can be computed much faster by the \emph{round-robin protocol} \cite{CKMPSW2019}: based on any fixed-in-advance order, the agents select the most preferable remaining item one by one until no item is left.
This procedure can be implemented to run in $\mathrm{O}(nm \log m)$ time, where the bottleneck is sorting the items for each agent.

This procedure, however, cannot be easily extended to nonadditive valuations.
If the agents select items one by one in a round-robin fashion so that each agent selects an item achieving the maximum increase of the value of his/her bundle, then the resultant allocation is not necessarily EF1; such an instance is easily constructed, e.g., by considering the situation when there are complementary items (i.e., the valuations are superadditive for some pair of items).
\end{remark}

\subsection{Improvement via Assignment Problem}\label{sec:improved2}
First, we slightly modify the definition of an allocation by regarding it as a combination of a ``division'' and an ``assignment.''
Hereafter,  an \emph{allocation}  is a pair $(\cX, \mu)$ of an $N$-indexed subpartition $\cX = (X_i)_{i \in N}$ of $M$ and a bijection $\mu \colon N \to \cX$, which means that bundle $\mu(i)\in \cX$ is assigned to agent $i\in N$.
We say that $(\cX, \mu)$ is \emph{EF1} or \emph{EFX} if this is true for subpartition $(X'_i)_{i \in N}$ with $X'_i = \mu(i)$ $(i \in N)$ in the sense of the original definition (see Definition~\ref{def:EF1}).
The \emph{welfare} of an allocation $(\cX, \mu)$ is defined as $W(\cX, \mu) = \sum_{i \in N} v_i(\mu(i))$.

\begin{definition}\label{def:assign-envy_graph}
For an allocation $(\cX, \mu)$ with the bundle set $\cX = (X_i)_{i \in N}$, the \emph{assign-envy graph} $G(\cX, \mu) = (N, \cX; E(\cX, \mu))$ is a bipartite graph with edge weights $w \colon E(\cX, \mu) \to \RR$ defined by
\begin{align*}
  E(\cX, \mu) &\coloneqq \{\, (i, \mu(i)) \mid i \in N \,\} \cup \{\, (i, X_{j}) \mid i\in N, ~X_{j} \in \cX,~v_{i}(\mu(i)) < v_{i}(X_j) \,\},\\
  w(e) &\coloneqq -v_{i}(X_j) \quad (e = (i, j) \in E(\cX, \mu)).
\end{align*}
\end{definition}

Note that the bijection $\mu\colon N\to \cX$ is regarded as a perfect matching in $G(\cX, \mu)$.\footnote{The negative sign appearing in the definition of edge weights is used for a clear correspondence to the assignment problem, in which the objective is to find a minimum-weight perfect matching whereas in the current problem an allocation with a large welfare is preferable.}
Also note that if we contract each edge $(i, \mu(i))$ in $G(\cX, \mu)$, the resulting graph will coincide with the envy graph $H(\cX')$ for $\cX' = (X'_i)_{i \in N}$ with $X'_i = \mu(i)$ $(i \in N)$, which has been defined in Section~\ref{sec:EFA}.
Thus, a cycle in $H(\cX')$ corresponds to a cycle in $G(\cX, \mu)_\mu$ and vice versa.

\begin{lemma}\label{lem:negative_cycle}
Let $(G, w)$ be the assign-envy graph for an allocation $(\cX, \mu)$, and $(G_\mu, w_\mu)$ be the auxiliary weighted graph for $\mu$.
If $G_\mu$ has a cycle $C$, then the weight $w_\mu(C)$ is negative.
\end{lemma}

\begin{proof}
If an allocation is updated by exchanging bundles along a cycle in the envy graph, each agent in the cycle receives a more valuable bundle; hence, the welfare strictly increases.
This is also true for the assign-envy graph.
By the definitions of $w$ and $(G_\mu, w_\mu)$, the weight $w_\mu(C)$ is equal to $W(\cX, \mu) - W(\cX, \mu') < 0$, where $\mu'$ is the assignment after the exchange along $C$.
\end{proof}

\begin{lemma}\label{lem:min-weight_perfect_matching}
Let $(G, w)$ be the assign-envy graph for an allocation $(\cX, \mu)$, and $\mu^*$ be a minimum-weight perfect matching in $(G, w)$.
Let $(G^*, w^*)\coloneqq G(\cX, \mu^*)$ and let $(G^*_{\mu^*}, w^*_{\mu^*})$ be the auxiliary weighted graph for $\mu^*$.
Then, the directed graph $G^*_{\mu^*}$ has no cycles.
\end{lemma}

\begin{proof}
Suppose to the contrary that $G^*_{\mu^*}$ has a cycle $C$.
We then have $w^*_{\mu^*}(C) < 0$ by Lemma~\ref{lem:negative_cycle}.
Let $\mu'$ be the assignment obtained from $\mu^*$ by the exchange along $C$.
We then have $W(\cX, \mu') > W(\cX, \mu^*)$ and $v_i(\mu'(i)) > v_i(\mu^*(i)) \ge v_i(\mu(i))$ for every $i \in N$ with $\mu'(i) \neq \mu^*(i)$.
The latter implies that all the edges in $\mu'$ exist in $G$, and then the former contradicts that $\mu^*$ is a minimum-weight perfect matching in $(G, w)$.
\end{proof}

By Lemma~\ref{lem:min-weight_perfect_matching}, instead of repeatedly exchanging the bundles along cycles in the envy graph, it suffices to find a minimum-weight perfect matching in the assign-envy graph once in each iteration of the envy-cycle procedure.
Note that if $(\cX, \mu)$ is an EF1 allocation, then $(\cX, \mu^*)$ is also EF1 since $v_i(\mu^*(i)) \ge v_i(\mu(i))$ for every $i \in N$ by the definition of the assign-envy graph.\footnote{Even if we start with the same EF1 allocation, the EF1 allocation obtained from a minimum-weight perfect matching may be different from that obtained by repeated exchanges in the original procedure, since the envy graph will change after each exchange.}
Thus, we obtain a variant of the envy-cycle procedure as shown in Algorithm~\ref{alg:ECP-AP}.

\begin{algorithm}[h]
	\caption{~ Envy-Cycle Procedure via Assignment Problem}
	\begin{description}
		\setlength{\itemsep}{-1mm}
		\item[Input:] A set $N$ of agents, a set $M$ of items, and valuations $v_{i}$ $(i \in N)$.
		\item[Output:] An EF1 allocation $(\cX, \mu)$ of $M$.
	\end{description}
	\begin{enumerate}
		\setlength{\itemsep}{0mm}
		\setlength{\leftskip}{-1mm}
		\item
		Initialize $\cX \leftarrow (X_i)_{i \in N}$ with $X_i = \emptyset$ for each $i \in N$, and $\mu(i) \leftarrow X_i$ for each $i \in N$.
		\item
		For each item $k \in M$ (in any order), do the following.\vspace{-1mm}
		\begin{enumerate}
			\setlength{\itemsep}{.5mm}
			\item
			Find a minimum-weight perfect matching $\mu^*$ in $(G(\cX, \mu), w)$, and update $\mu \leftarrow \mu^*$.
			\item
			Pick a bundle $X_i \in \cX$ that has only one incident edge in $G(\cX, \mu)$ (which is $(\mu^{-1}(X_i), X_i) \in \mu$), and update $X_i \leftarrow X_i \cup \{k\}$.
		\end{enumerate}
		\item
		Return $\cX$.
	\end{enumerate}
	\label{alg:ECP-AP}
\end{algorithm}

The number of edges in $G(\cX, \mu)$ is always bounded by $\mathrm{O}(n^2)$.
Hence, using the Hungarian method in Step~2(a), this algorithm correctly finds an EF1 allocation in $\mathrm{O}(mn^3)$ time, which is the same bound as Theorem~\ref{thm:Lipton}.

We improve this bound with the aid of Algorithm~\ref{alg:PDU} given in Section~\ref{sec:balanced}.
We deal with $G(\cX, \mu)$ as the complete bipartite graph by adding all the absent edges with sufficiently large weights.
Then, in each iteration of Step~2, only the weights of edges around bundle $X_i \in \cX$ picked in Step~2(b) change; hence, we can correctly update the current assignment $\mu$ and the corresponding potential in $\mathrm{O}(n^2)$ time by Algorithm~\ref{alg:PDU}.
Thus, the total computational time is bounded by $\mathrm{O}(mn^2)$, leading to the following theorems, which are analogous to Theorems~\ref{thm:Lipton} and \ref{thm:Lipton2}.

\begin{theorem}\label{thm:Lipton_improved}
Algorithm~\ref{alg:ECP-AP} correctly finds an EF1 allocation in $\mathrm{O}(mn^2)$ time.
\end{theorem}

\begin{theorem}\label{thm:Lipton2_improved}
Suppose that the valuations of all agents are additive, and $v_{i,1} \ge v_{i,2} \ge \dots \ge v_{i,m}$ for every agent $i \in N$.
Then, if the items are chosen in ascending order of $k = 1, 2, \dots, m$ in Step~2, Algorithm~\ref{alg:ECP-AP} correctly finds an EFX allocation in $\mathrm{O}(mn^2)$ time.
\end{theorem}

\begin{remark}
The computational time is actually bounded by using the number of edges in $G(\cX, \mu)$, which may be $\Omega(n^2)$ in the worst case.
In practice, however, the algorithm runs much faster than this theoretical bound, which is also true for the original envy-cycle procedure, because an acyclic graph tends to be sparse.
In addition, in this special situation, without keeping a potential, we can update the current assignment $\mu$ by solving the shortest path problem in directed acyclic graphs as follows; this slightly improves the running time when $G(\cX, \mu)$ is almost always sparse, in particular, with $\mathrm{o}(n \log n)$ edges.

Let $(G, w)$ and $(G', w')$ be the assign-envy graphs for the allocations $(\cX, \mu)$ just before and after Step~2(b), respectively.
As $\mu$ is a minimum-weight perfect matching in $(G, w)$, the auxiliary weighted graph $(G_\mu, w_\mu)$ has no cycles by Lemma~\ref{lem:negative_cycle}.
In addition, since $(G', w')$ coincides with $(G, w)$ except around $X_i\in \cX$ picked in Step~2(b), every cycle in $G'_\mu$ intersects $X_i$.
This means that the auxiliary weighted graph $(G_{\mu'}, w_{\mu'})$ to which we apply Dijkstra's algorithm in Step~2 of Algorithm~\ref{alg:PDU} is always acyclic, where $\mu' = \mu \setminus \{(\mu^{-1}(X_i), X_i)\}$.
Thus, Dijkstra's algorithm with a potential is not necessary, and the elementary, linear-time dynamic programming is sufficient for computing a shortest path from $\mu^{-1}(X_i)$ to $X_i$.
\end{remark}